\documentclass{article}
\usepackage{algorithm,algcompatible}
\usepackage{listings}
\usepackage{amsmath, amssymb, amsthm}
\usepackage{graphicx}
\usepackage{hyperref}
\usepackage[round]{natbib}

\newcommand{\CS}{\mathbb{S}}
\newcommand{\CW}{\mathbb{W}}
\newcommand{\CC}{\mathbb{C}}
\newcommand{\bX}{\mathbf{X}}
\newtheorem{theorem}{Theorem}
\newtheorem{corollary}{Corollary}

\begin{document}

\title{Assessing Web Fingerprinting Risk\footnote{A version of this report to appear in the proceedings of The Web Conference (WWW) 2024. This version contains additional material in the appendix.}}
\author{Enrico Bacis \and Igor Bilogrevic \and Robert Busa-Fekete \and Asanka Herath \and Antonio Sartori \and Umar Syed}
\date{Google}

\maketitle

\begin{abstract}
Modern Web APIs allow developers to provide extensively customized experiences for website visitors, but the richness of the device information they provide also make them vulnerable to being abused to construct browser \emph{fingerprints}, device-specific identifiers that enable covert tracking of users even when cookies are disabled.

Previous research has established \emph{entropy}, a measure of information, as the key metric for quantifying fingerprinting risk. However, earlier studies had two major limitations. First, their entropy estimates were based on either a single website or a very small sample of devices. Second, they did not adequately consider correlations among different Web APIs, potentially grossly overestimating their fingerprinting risk.

We provide the first study of browser fingerprinting which addresses the limitations of prior work. Our study is based on actual visited pages and Web APIs reported by tens of millions of real Chrome browsers in-the-wild. We accounted for the dependencies and correlations among Web APIs, which is crucial for obtaining more realistic entropy estimates. We also developed a novel experimental design that accurately and efficiently estimates entropy while never observing too much information from any single user. Our results provide an understanding of the distribution of entropy for different website categories, confirm the utility of entropy as a fingerprinting proxy, and offer a method for evaluating browser enhancements which are intended to mitigate fingerprinting.
\end{abstract}

\section{Introduction}\label{section:introduction}

Modern websites make extensive use of Web APIs to provide sophisticated functionality to their users. Some of these APIs reveal information about the user's device that developers can use to tailor the browsing experience for individual users. For example, a document editing website might only render the fonts that are installed on a user’s system, or a video conferencing website might leverage a device’s graphics card to accelerate performance.
However, these granular details can also be used to generate browser \emph{fingerprints}, which are (near) unique identifiers linked to a specific device. Such fingerprints can be used to track users across many different websites, often without their knowledge, and thus pose a major privacy risk. Unlike cookies, fingerprints cannot be easily disabled or reset by users. Although browser fingerprinting can be used to enhance users' online safety (such as preventing fraud or abuse), it is generally perceived as intrusive and therefore discouraged by several popular browsers (e.g., Safari, Firefox, Edge and Brave). 

To better assess the risk of browser fingerprinting for users and the effectiveness of anti-fingerprinting methods of browsers, we conducted a large-scale measurement study of Web APIs that focused on quantifying their \emph{entropy}, a measure of information. Each Web API exposes users to a degree of fingerprinting risk that depends on the distribution of its values in the user population. A Web API whose value distribution is very concentrated on a few possible values (Figure~\ref{figure:distributions}(a)) provides less identifying information about the average user than a function whose value distribution is very spread out (Figure~\ref{figure:distributions}(b)), since sharing values with many users is better for preserving anonymity. Because entropy increases with the spread of a distribution, it serves as an indicator of how well-suited a Web API is for re-identifying users, and has been widely used by researchers as a key metric for quantifying fingerprinting risk \citep{eckersley2010unique,Laperdrix2020-yj}.

Previous entropy measurement studies suffered from two major limitations. They were either limited to a small sample of websites or devices, and/or they ignored potential correlations among Web APIs, which made it challenging for them to provide realistic estimates of fingerprinting risk.
Our novel experimental design allowed us to collect data about dozens of Web APIs called by tens of millions of Chrome browser instances on hundreds of thousands of websites, while limiting the information collected from any single user, thus protecting user privacy. We estimated the correlation between many pairs of Web APIs, enabling us to not only estimate the fingerprinting risk posed by individual Web APIs, but also the joint risk posed by sets of them. Due to the privacy protection mechanisms put in place by our study, we were able to observe telemetry from real user devices instead of an automated crawler, reflecting the real-world behavior of websites as opposed to artificial responses to a bot.

Our results confirmed that Web API usage patterns vary significantly across website categories, and that the information revealed by many Web APIs is often highly correlated. We leveraged these findings to produce the broadest and most realistic measurement of the entropy of Web APIs to date. We validated the use of entropy as a proxy for fingerprinting by showing that it is strongly correlated with known fingerprinting behavior. We also studied how methods that are designed to mitigate fingerprinting can impact the web by quantifying how much information they prevent from being exposed to websites.

\begin{figure}[t]
  \centering
  \includegraphics[width=\linewidth, trim={0 1cm 0 0}]{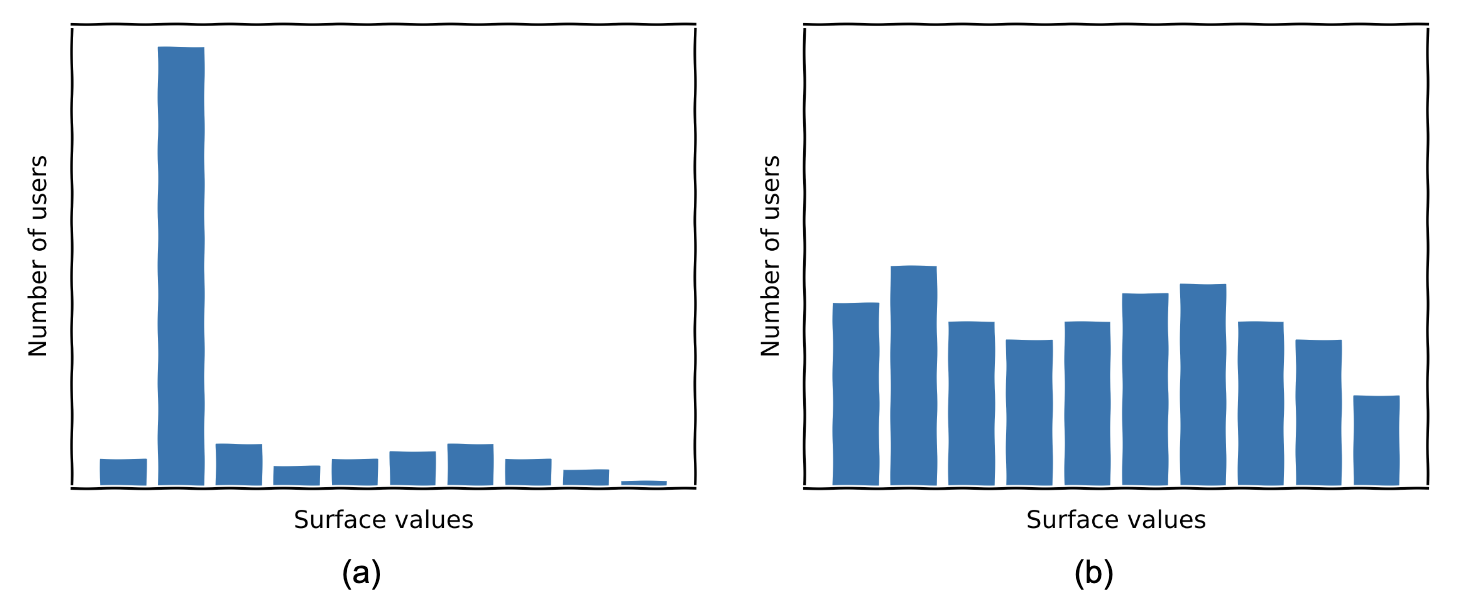}
  \caption{Two possible value distributions for a Web API, (a) low entropy and (b) high entropy.}
  \label{figure:distributions}
\end{figure}

\section{Related Work}\label{sec:rel_work}
Browser fingerprinting has attracted a vast amount of research over the past 15 years~\citep{Laperdrix2020-yj,zhang2022survey}. From a broad perspective, prior works that are most related to our study can be grouped in two categories, depending on their contributions:
\begin{enumerate}
    \item Large-scale measurement studies on browser and attribute uniqueness~\citep{eckersley2010unique,laperdrix2016beauty,gomez2018hiding,Andriamilanto2021-og}.
    \item Attribute set selection studies and inter-attribute correlation~\citep{flood2012browser,fifield2015fingerprinting,andriamilanto2020fpselect,Tanabe2019-pk,andriamilanto2021large,pugliese2020long}.
\end{enumerate}

When studying browser uniqueness, it is important to have access to a large and diverse dataset in order to have as many samples from the tail-end of the distribution as possible. These infrequent and uncommon values usually carry the highest amount of identifiable information, and they can have a significant impact on the study results. Therefore, we only discuss prior works that are based on large-scale experiments with real browsers visiting public websites.

Furthermore, as this work focuses exclusively on uniqueness and fingerprinting, we do not discuss the broader topic of online tracking, which is often associated with fingerprinting. A survey on that topic can be found in~\citep{bujlow2017survey}.

\subsection{Browser uniqueness}

Considering the vast amount of different configurations of hardware and software components of Internet-connected devices, it is unsurprising to find some configurations to be (nearly) unique. That was the underlying premise of the initial studies of browser uniqueness as observed through web APIs~\citep{Laperdrix2020-yj}. To date, there are four large-scale studies that attempted to measure browser uniqueness in-the-wild~\citep{eckersley2010unique,laperdrix2016beauty,gomez2018hiding,Andriamilanto2021-og}. 

In his pioneering work, \cite{eckersley2010unique} created a test site\footnote{The initial website was \url{https://panopticlick.eff.org}, which in 2017 moved under \url{https://coveryourtracks.eff.org}.} where visitors could choose to have their browser analysed and fingerprinted via HTTP headers, JavaScript and plugins. The two-week experiment gathered over 470k fingerprints, 83.6\% of which were unique. This study, which was the first to use the expression ``browser fingerprinting'', was also the first to use information entropy (in terms of ``surprisal'') as a proxy for privacy risk on the web, and also among the first to quantify the evolution of fingerprints over time (see~\cite{li2020touched} for more recent results on fingerprints evolution over time). They estimated that a browser (in 2010) carried at least 18.1 bits of entropy, with the top-3 properties contributing with (individually) 15.4 bits (plugins), 13.9 bits (fonts), and 10 bits (user agent).

\cite{laperdrix2016beauty} conducted an analysis of over 118k fingerprints over a 1-year period, by deploying a test website,\footnote{\url{https://amiunique.org}} similar to Eckersley's approach from 6 years before. While, on the one hand, the results confirmed the prior findings as 89.4\% of fingerprints were unique, on the other hand they also found significant differences in the prevalence and type of fingerprinting techniques. In particular, the authors found that the availability of the HTML5 canvas element and the wide browser support for webGL APIs contributed significantly to the amount of uniqueness, especially on mobile devices, in the absence of plugins and Flash.

Contrary to the previous two studies~\citep{eckersley2010unique,laperdrix2016beauty}, which created a dedicated test website where visitors could choose to have their fingerprints collected, \cite{gomez2018hiding} collected more than 2 million fingerprints on one of the top-15 most visited French websites. In contrast to the prior results, their results indicated that only 33.6\% of all fingerprints were unique, highlighting the potentially large impact that the data collection methodology can have on the findings. \cite{Andriamilanto2021-og} conducted a 6-month study where they collected more than 4 million fingerprints from almost 2 million visitors to the same French website as~\cite{gomez2018hiding}. However, this time they collected 216 attributes (as opposed to 17) from a significantly larger user population. In line with prior results~\citep{eckersley2010unique,laperdrix2016beauty}, the authors observed that 81\% of all fingerprints are unique, and more than 90\% of the attributes remained unchanged even after 6 months.

While the four studies discussed so far have provided important insights into browser uniqueness, the strength of their results was limited; they only studied uniqueness broadly (i.e., uniqueness of the full fingerprints) or uniqueness of individual attributes (e.g., user agent strings). Yet, different device properties (both hardware and software) are often correlated with each other~\citep{fifield2015fingerprinting,andriamilanto2020fpselect,pugliese2020long}, which makes it crucial to consider them not in \emph{isolation} but \emph{jointly}, as discussed hereafter.

\subsection{Attribute set selection}
\cite{fifield2015fingerprinting} devised a greedy attribute set selection heuristic that considered the conditional entropy when adding subsequent attributes to a set, in order to obtain the smallest cardinality with the highest fingerprinting potential. Although it provided interesting results, the study was limited in two ways. First, it was based on a very small sample of 1,016 participants from Mechanical Turk. Second, the algorithm only worked on fonts attributes. 

\cite{andriamilanto2020fpselect} framed the attribute set selection problem as an optimization problem with two joint objectives: distinctiveness and usability cost. Their greedy algorithm explores a given number of potential sets, by adding attributes one by one to each such set iteratively (similarly to~\cite{fifield2015fingerprinting}), and retaining the sets that satisfy the constraints. Although this study used the same large dataset as~\cite{gomez2018hiding}, it still only considered sets that were constructed by adding individual attributes one by one.

\cite{pugliese2020long} conducted a study with 1,304 users and collected 305 attributes, resulting in 88,088 measurements. Their greedy attribute set selection method was based on maximising the stability of trackable attributes as well as the number of participants with trackable attributes, resulting in different sets for mobile and desktop devices. However, the algorithm added attributed one-by-one from the small dataset, and recomputed the objective function until no further improvement can be achieved.

In contrast to related work, our study has a much larger dataset, collected from significantly more websites, a more accurate joint entropy estimation algorithm and better privacy protection. In particular, our study includes data from more than 10 million different Chrome clients who visited more than 100,000 websites in-the-wild. This represents 10x more participants as compared to the other largest study to date~\citep{Andriamilanto2021-og}, which was done on a single website.

\section{Notation and terminology}
\label{sec:notation}

The web consists of a set of websites $\CW$ and the web population consists of a set of clients $\CC$, each of which is a browser installation on a single device. A \emph{surface} is a Web API combined with a particular set of input arguments. For example, the Javascript function {\tt WebGLRenderingContext.getParameter(pname)} returns a parameter value of a client's graphics environment, where {\tt pname} specifies the parameter. Let $\CS$ be the set of all possible surfaces, and let $x_{s, c}$ be the return value of surface $s$ for client $c$. Website $w$ observes surface value $x_{s, c}$ if and only if client $c$ visits website $w$ and the site executes a script that invokes the Web API with input arguments corresponding to surface $s$. 

Let client $C$ be chosen uniformly at random from $\CC$. For any surface $s \in \CS$ define the random variable $X_s = x_{s, C}$, and for any surface set $S = \{s_1, \ldots, s_k\} \subseteq \CS$ define the random variable $\bX_S = (X_{s_1}, \ldots, X_{s_k})$. In other words, the distribution of $\bX_S$ is the distribution of joint values for surfaces $S$ in the client population.

A website can extract different pieces of identifying information from a device over multiple subsequent visits, and therefore it is important to account for this accumulation of information when estimating fingerprinting risk. We define a \emph{session} to be the aggregation of all visits to a specific website by a specific client over a 4 week period.

\section{Efficient entropy estimation}\label{sec:entropy}

The entropy of a discrete random variable $X$ is defined
\[
H(X) = -\sum_x \Pr[X = x] \log \Pr[X = x].
\]
For any surface set $S \subseteq \CS$, the entropy $H(X_S)$ quantifies the fingerprinting risk to users who reveal their values for surfaces in $S$ to websites. The primary objective of our measurement study is to estimate $H(X_S)$ for various surface sets $S$, a goal that presents at least two major challenges. First, the number of possible surfaces sets is $2^{|\CS|}$, and so it is not practical to form an independent estimate for each surface set. Second, for a given surface set $S$, the number of samples from $X_S$ needed to estimate $H(X_S)$ can be prohibitive. For example, if $|S| = k$, then even if each surface in $S$ has only two possible values, the optimal entropy estimation algorithm needs $\Omega(\frac{2^k}{k})$ samples to estimate $H(X_S)$ accurately \citep{wu2016minimax}.

Instead of directly estimating each $H(X_S)$, we estimate an upper bound based on a well-known decomposition due to \citet{chow1968approximating}. Let $T$ be any tree whose nodes are the surfaces $S$. \citeauthor{chow1968approximating} showed that  
\begin{equation}
H(X_S) \le \sum_{s \in S} H(X_s) - \sum_{(s, s') \in T} I(X_s; X_{s'}) \label{eq:chowliu}
\end{equation}
where the second sum is over the edges of $T$ and the mutual information $I(X; Y)$ between random variables $X$ and $Y$ is defined
\[
I(X; Y) = \sum_{x, y} \Pr[X = x, Y = y] \log \frac{\Pr[X = x, Y = y]}{\Pr[X = x]\Pr[Y = y]}.
\]

Mutual information, which is always non-negative, can be viewed as a measure of the correlation between two random variables. If $X$ and $Y$ are highly correlated then $I(X; Y)$ is large, and if they are independent then $I(X; Y)$ is zero. Thus when surfaces are correlated the Chow-Liu decomposition gives a much tighter upper bound on entropy than the more naive bound $H(X_S) \le \sum_{s \in S} H(X_s)$, which has been used in previous work \citep{mendes2011nophish}. For example, our measurement study found that the Javascript Web APIs that return a user's language setting and keyboard layout have high mutual information (since keyboard layouts are often designed for a particular language), so their joint entropy is significantly lower than the sum of their individual entropies.

At most $|\CS|^2$ entropy and mutual information terms must be estimated in order to compute the Chow-Liu upper bound in Eq.~\eqref{eq:chowliu} for all subsets of $\CS$, which is much less than $2^{|\CS|}$. As we explain in Section \ref{sec:methodology}, we make the problem even more tractable by not estimating the mutual information between every pair of surfaces, but instead judiciously choosing surface pairs that we expect to be correlated. Finally, each term in the Chow-Liu upper bound can be estimated efficiently. If each surface has $O(v)$ possible values, then the entropy terms need $O(v)$ independent samples to estimate accurately, and the mutual information terms need $O(v^2)$ samples (see Appendix \ref{sec:samplecomplexity} for a more precise statement of the sample requirements).

\section{Experimental methodology}
\label{sec:methodology}

We constructed the set of all surfaces $\CS$ by considering all Web APIs exposed via either JavaScript or CSS, manually selecting surfaces that we believed could be sources of stable entropy, while excluding APIs which are gated behind a permission prompt. This led to $\CS$ containing a total of 5383 surfaces derived from 161 Web APIs (see Appendix \ref{sec:webapis} for complete list of the Web APIs). Note that this approach placed some potential sources of entropy outside the scope of our experiment, such as any information exposed by installed Chrome extensions.

Our experimental goal was to estimate the entropy $H(X_S)$ of many different surface subsets $S \subseteq \CS$ subject to the following constraints:
\begin{itemize}
    \item {\bf Privacy:} Our experiment did not collect too much information from any single user.
    \item {\bf Performance:} Our experiment did not interfere with the proper rendering of websites or unduly slow the browser.
\end{itemize}

Given these constraints, we made several key choices when designing our experimental methodology.
\begin{itemize}
    \item We approximated the entropy $H(X_S)$ of a surface set $S \subseteq \CS$ by using the Chow-Liu decomposition which, as we discussed in Section \ref{sec:entropy}, reduces overall sample requirements by an exponential factor.
    \item We took a `passive' approach to data collection, by only recording those surface values that were actually observed by websites during client sessions, rather than a more `active' approach that would execute Web APIs on the client's behalf even if they were not called by any script. Passive sampling prevented performance degradation or unexpected browser behavior due to Web API side effects, and ensured that a user's experience of a website was exactly as the web developer intended.
    \item We kept the client-side state and logic associated with our experiment as minimal as possible. Each client $c \in \CC$ stored only a list of candidate surfaces $L_c \subseteq \CS$, a reporting probability $p_c \in [0, 1]$, and the number of reported surface values $k_c$. Whenever client $c$ visited a website that observed a value for a surface in $L_c$, with probability $p_c$ that surface value was reported to our servers and the value of $k_c$ was incremented. Surface values were briefly cached on the client before reporting, but were not stored for an extended duration.
    \item We preserved user privacy by hashing surface values on the client before reporting them, discarded any surface values reported by fewer than 50 clients, and deleted all unaggregated data after 4 weeks. Note that hashing had no impact on entropy estimation accuracy, since entropy is invariant to renaming the surface values. We also never collected more than 20 bits of device-specific information from any client (we explain below how this requirement was enforced). Only a small fraction of Chrome clients were included in our experiment, accounting for less than 10 million clients. Clients had to have the following two settings enabled in Chrome to be included in the experiment: (1) `Share usage reports and crash analytics with Google', and (2) `Make searches and browsing better / Sends URLs of pages you visit to Google'~\citep{chrome-white-paper}.
\end{itemize}

We conducted the experiment in several phases, with the results of each phase informing the design of subsequent phases:

\begin{itemize}
    \item {\bf Phase 1:} In order to obtain a high-level understanding of Web API usage, we set $L_c = \CS$ and $p_c = 1$ for each client $c$, but only reported which surfaces were called on the client, and not the surface values (not even hashed versions of the values). No device-specific information was collected during this phase. 
    \item {\bf Phase 2:} We partitioned the surfaces $\CS$ into 10 large ``families'' (see Section \ref{section:call-frequency}) based on which surfaces we expected to be correlated with each other. We set $L_c$ for each client $c$ to one of the 10 families, thereby increasing the likelihood of observing joint samples from pairs of surfaces from the same family, which in turn led to tighter estimates of the Chow-Liu upper bound (Eq.~\eqref{eq:chowliu}). We also set the reporting probability $p_c = \min\{1, \frac{40}{n}\}$, where $n$ is the expected number of surfaces in client $c$'s assigned family that will be called during the experiment, a value determined by the data collected in Phase 1. However, once $k_c$ (the number of surface values reported by client $c$) reached 40, the value of $p_c$ was set to $0$. The effect of choosing these values for $p_c$ is that the typical client reported an unbiased random sample of their surface values, but no client reported more than 40 surface values. Note that setting a significantly higher initial value for $p_c$ would have skewed the reports towards surfaces that are called earlier in the experiment (because of the 40 report limit per client) and biased the data collection. We conservatively choose a 40 report limit during this phase based on an earlier, much smaller and unpublished experiment which estimated that the average surface exposes 0.5 bits of entropy. 
    \item {\bf Phase 3:} While the previous phase collected enough data to estimate the entropy of each individual surface, we found that it was insufficient to estimate the mutual information between many pairs of surfaces (particularly surfaces from two different families), essentially because the probability of jointly observing a pair of surface values from the same client $c$ is only $p_c^2$. So we augmented our data collection with an additional sampling phase that focused on popular Web APIs. We identified a set of 60 popular surfaces (the `core', see Section \ref{section:clustering}) and created many overlapping subsets of these surfaces, with each subset constructed so that its entropy was expected to be low, based on the data collected in Phase 2. For each client $c$ we let $L_c$ be one of these subsets and let $p_c = 1$. We estimated that no client reported more than 20 bits of information in Phase 3. Note that the clients used in this phase were disjoint from the clients used in the previous phase.
\end{itemize}

We used the individual surface entropy estimates and pairwise mutual information estimates computed from the data collected in Phases 2 and 3 to calculate an upper bound on the entropy $H(X_S)$ for many surface sets $S \subseteq \CS$, based the Chow-Liu decomposition. When forming the decomposition, we omitted mutual information terms with insufficient samples, as determined by the derivations in Appendix \ref{sec:samplecomplexity} (note that omitting mutual information terms does not impact the validity of the upper bound in Eq.~\eqref{eq:chowliu}). Our final estimates of an upper bound for each $H(X_S)$ were accurate to within 1.5 bits with 90\% confidence.

Some aspects of our experimental methodology have been omitted from the description in this section for the sake clarity. Complete details can be found in Appendices \ref{sec:observable}, \ref{sec:special}, \ref{sec:greedy} and \ref{sec:monotonicity}.

\section{Results}\label{section:results}

\begin{figure}[!h]
  \centering
  \includegraphics[width=0.5\linewidth, trim={0 1.5cm 0 0}]{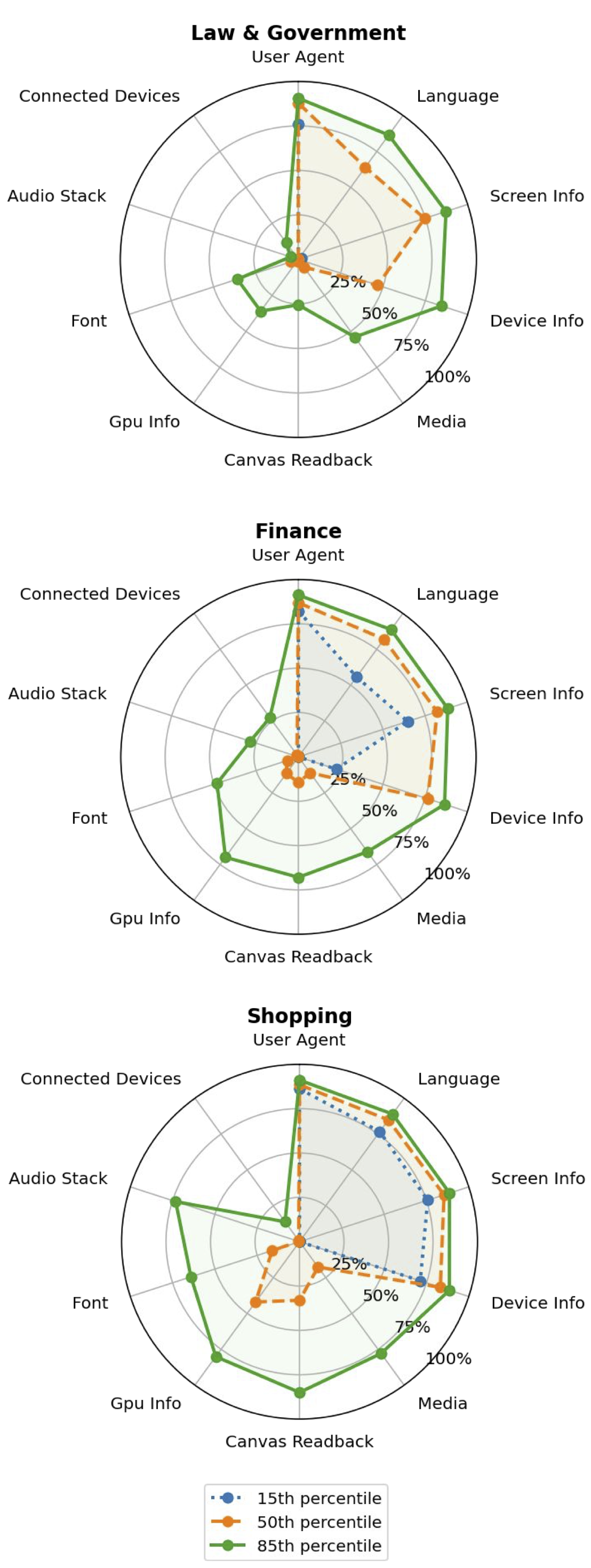}
  \caption{Surface call frequency for several website verticals}
  \label{figure:radars}
\end{figure}

The findings of our measurement study can be grouped into several levels of analysis. First, we calculated how often different families of surfaces (fonts, media, etc) are called by websites, which provided a high-level understanding of how Web API functionality is used across the web. Next, we estimated the entropy of individual surfaces, and the correlation and mutual information between pairs of surfaces, and clustered together surfaces that expose similar identifying information. Finally, we estimated the joint entropy of all the surfaces observed during each client session on each website. This allowed us to plot the entropy distribution for all sessions, validate entropy as a fingerprinting metric, and evaluate the effectiveness of anti-fingerprinting methods based on blocking scripts.

\subsection{Surface call frequency}\label{section:call-frequency}

We partitioned all surfaces into 10 families, and mapped each website to a vertical using a procedure similar to the one used for the Topics API \citep{topics}. We defined, for each surface family $f$, a threshold $t_f$, and we calculated the percentage of sessions in which at least $t_f$ different surfaces from surface family $f$ was accessed on a website (recall the definition of a session from Section \ref{sec:notation}). In the radar charts in Figure~\ref{figure:radars} we plot several percentiles of these percentages for a few website verticals. In each of these charts, the color indicates the percentage of websites, while the radial dimension indicates the percentage of sessions. For example, Figure~\ref{figure:radars}(c) shows that 15\% of shopping websites access screen information in 75\% of sessions.

Observe that website behavior can significantly vary across verticals. For example, GPU information is much more commonly accessed by shopping websites than law \& government websites, while nearly all websites access information about the user agent.

\subsection{Surface clustering}\label{section:clustering}

Generally speaking, users are exposed to fingerprinting risk whenever a website calls many high entropy surfaces during a browsing session. But this risk is reduced if the surfaces are correlated with each other, since in that case their joint entropy will be less than the sum of their individual entropies (see Eq.~\eqref{eq:chowliu}). We calculated all pairwise correlations between 60 `core' surfaces that were selected for their popularity, stability and non-triviality. Specifically, a surface was deemed to be a core surface if over a 1 month period: (1) it was called by at least 500 clients, (2) at least 95\% of clients reported only one value for the surface, (3) the surface's entropy was at least 0.1 bits, and (4) the number of clients reporting the surface did not vary significantly each day. The last requirement was intended to filter out surfaces that are only ephemerally popular, such as a canvas drawing function that is temporarily invoked with an unusual set of arguments by a popular website (recall from Section \ref{sec:notation} that a surface is a Web API  combined with its input arguments).

Defining the correlation between two discrete-valued random variables, such as a pair of surfaces, is not as straightforward as for real-valued random variables, for which metrics like Pearson correlation are typically used. Our approach is based on the observation that two random variables are independent if and only if their joint distribution is equal to the product of their marginal distributions. We applied a method from \citet{chan2014optimal} to estimate the distance between these distributions, as measured by total variation. Figure \ref{figure:clusters} reports the results for all pairs of core surfaces, with black squares indicating correlation (total variation is at least 0.05), red squares indicating independence (total variation is zero), and an empty square if there is not enough data to distinguish between these two cases with confidence at least 90\%. The entropy of each surface is shown on the right-hand side of Figure~\ref{figure:clusters}. An advantage of our definition of correlation is that the sample requirements for estimating it are fairly modest.

\begin{figure}[t]
  \centering
  \includegraphics[width=0.9\linewidth, trim={0 1.5cm 0 0}]{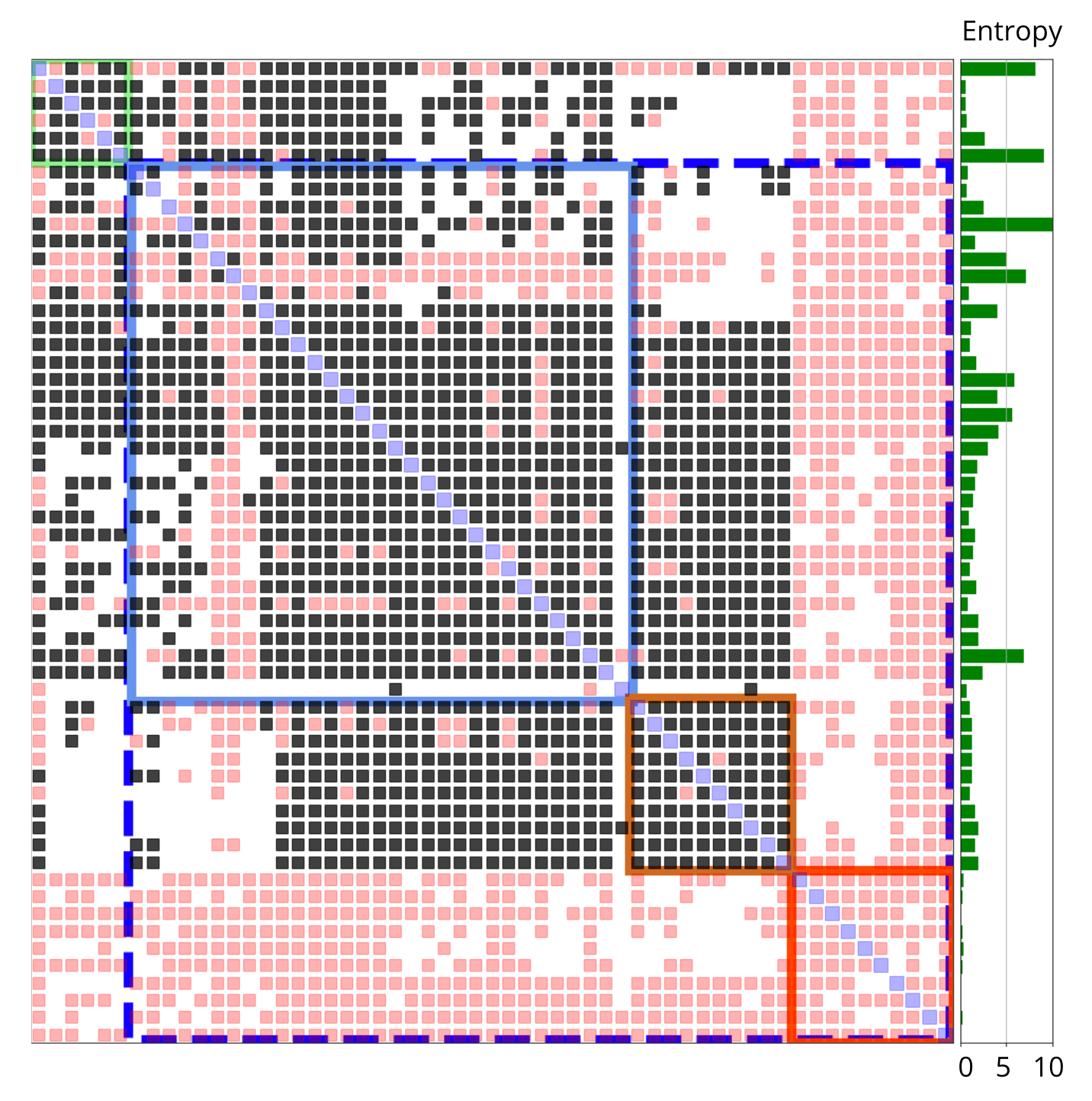}
  \caption{Surface clustering determined by pairwise correlations. Clusters are outlined by colored rectangles.}
  \label{figure:clusters}
\end{figure}

We applied the Cuthill–McKee algorithm \citep{cuthill1969reducing} to reorder the surfaces so that correlated surfaces tend to be adjacent to one another in the figure, and after grouping adjacent and related surfaces together, the following surface clustering emerged:

\begin{itemize}
    \item {\bf Device properties} \textit{(dashed blue rectangle)}. This cluster includes basic parameters of the device, such as the screen width and height, and is divided into three sub-clusters:
    \begin{itemize}
        \item {\bf Core device properties} \textit{(light blue rectangle)}. Includes language, pixel ratio, and screen width and height.
        \item {\bf WebGL parameters} \textit{(brown rectangle)}. Includes the WebGL parameters.
        \item {\bf Independent properties} \textit{(red rectangle)}. Although these surfaces are not correlated with any other surface, they have very low entropy, and so they pose minimal fingerprinting risk.
    \end{itemize}
    \item {\bf User agent cluster} \textit{(green rectangle)}. Surfaces which return information about the user agent. This cluster also has some cross-cluster correlation with device properties.
\end{itemize}

Expanding this analysis to all surfaces also reveals a natural clustering. Figure~\ref{figure:tree} is a single-linkage clustering \citep{everitt2011cluster} of all the surfaces. In the figure, each node represents a surface, a node’s size is proportional to the surface’s entropy, and the weight of the edge between two nodes is the mutual information between the surfaces. Observing many surfaces from the same cluster will, to a first approximation, expose a user to less fingerprinting risk than observing the same number of surfaces from different clusters.

\begin{figure}[t]
  \centering
  \includegraphics[width=0.9\linewidth, trim={0 1.5cm 0 0}]{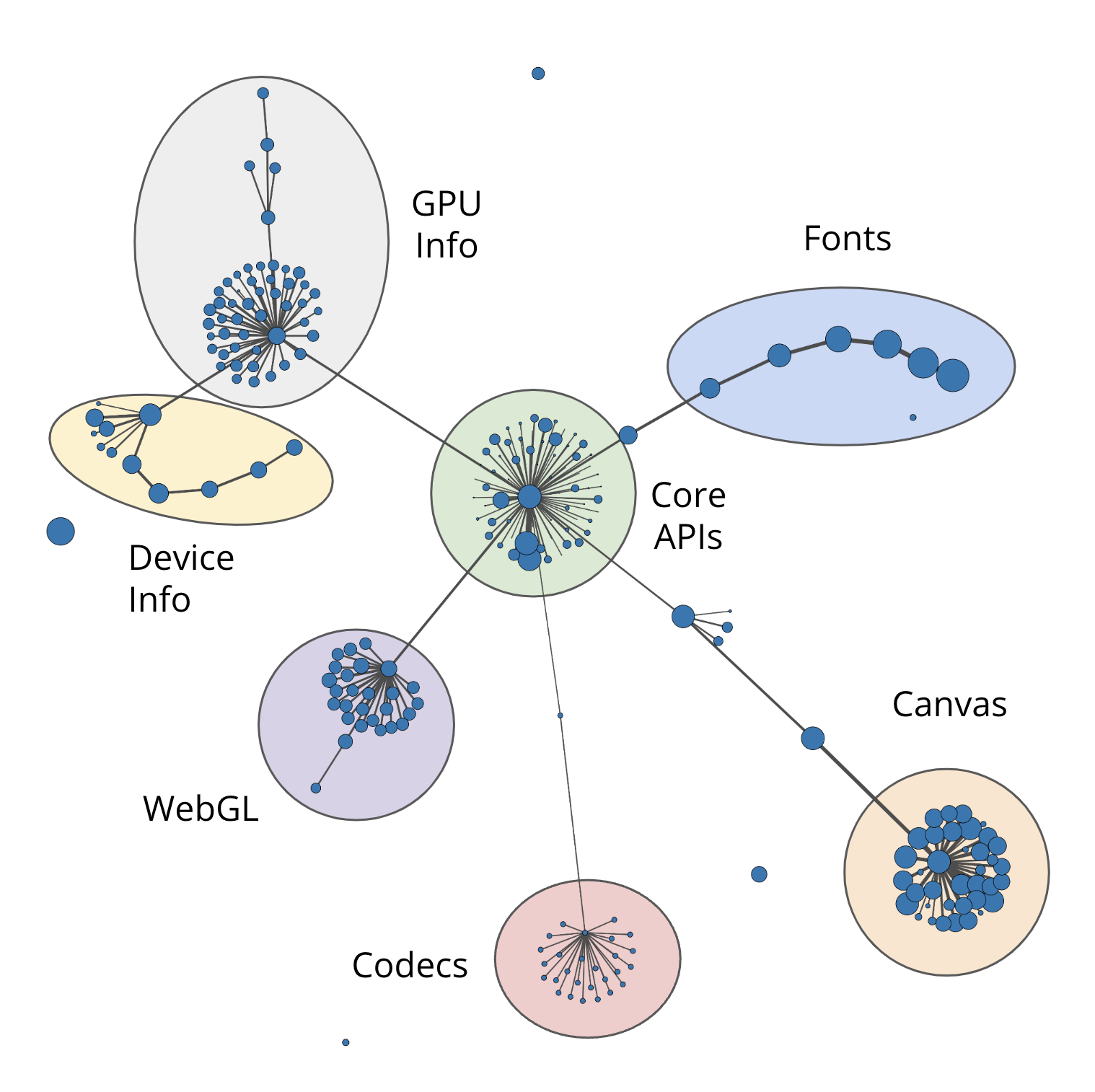}
  \caption{Surface clustering determined by mutual information. Clusters are indicated by ovals.}
  \label{figure:tree}
\end{figure}

\subsection{Entropy of sessions}\label{section:joint-entropy}

\begin{figure*}[t]
  \centering
  \includegraphics[width=0.8\linewidth, trim={0 1.5cm 0 0}]{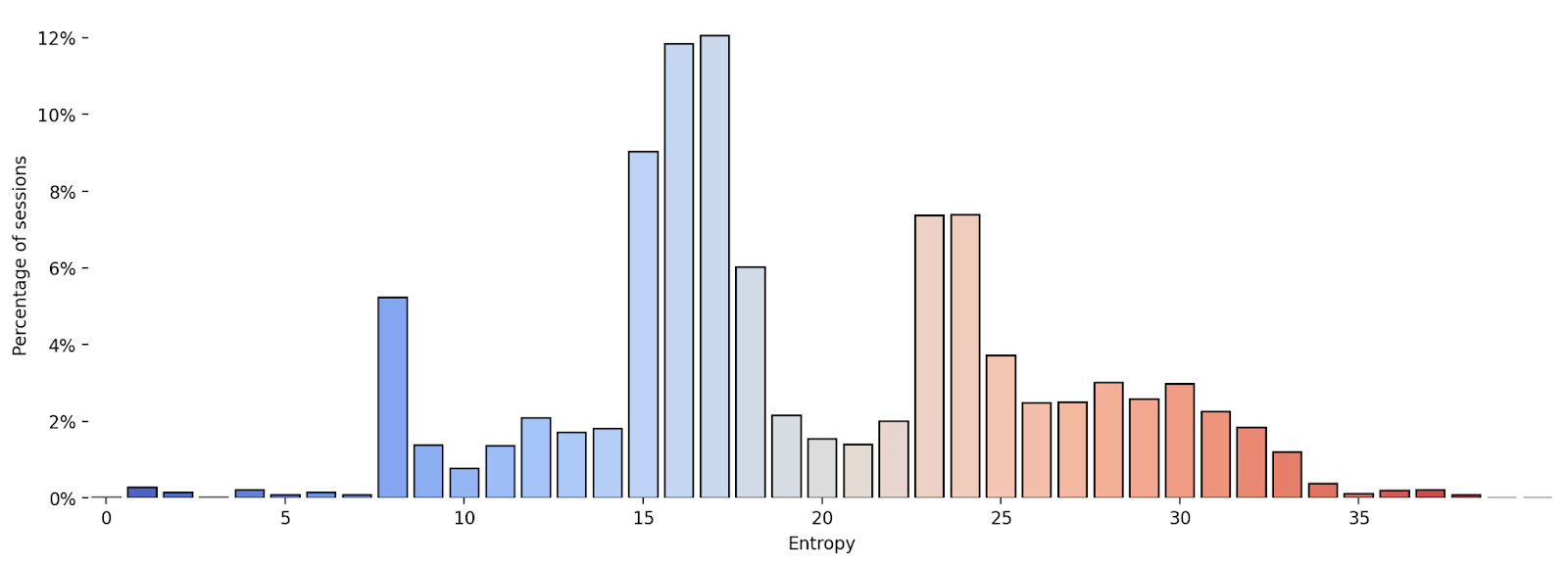}
  \caption{Entropy distribution of the web}
  \label{figure:web}
\end{figure*}

\begin{figure}[t]
  \centering
  \includegraphics[width=0.9\linewidth, trim={0 1.5cm 0 0}]{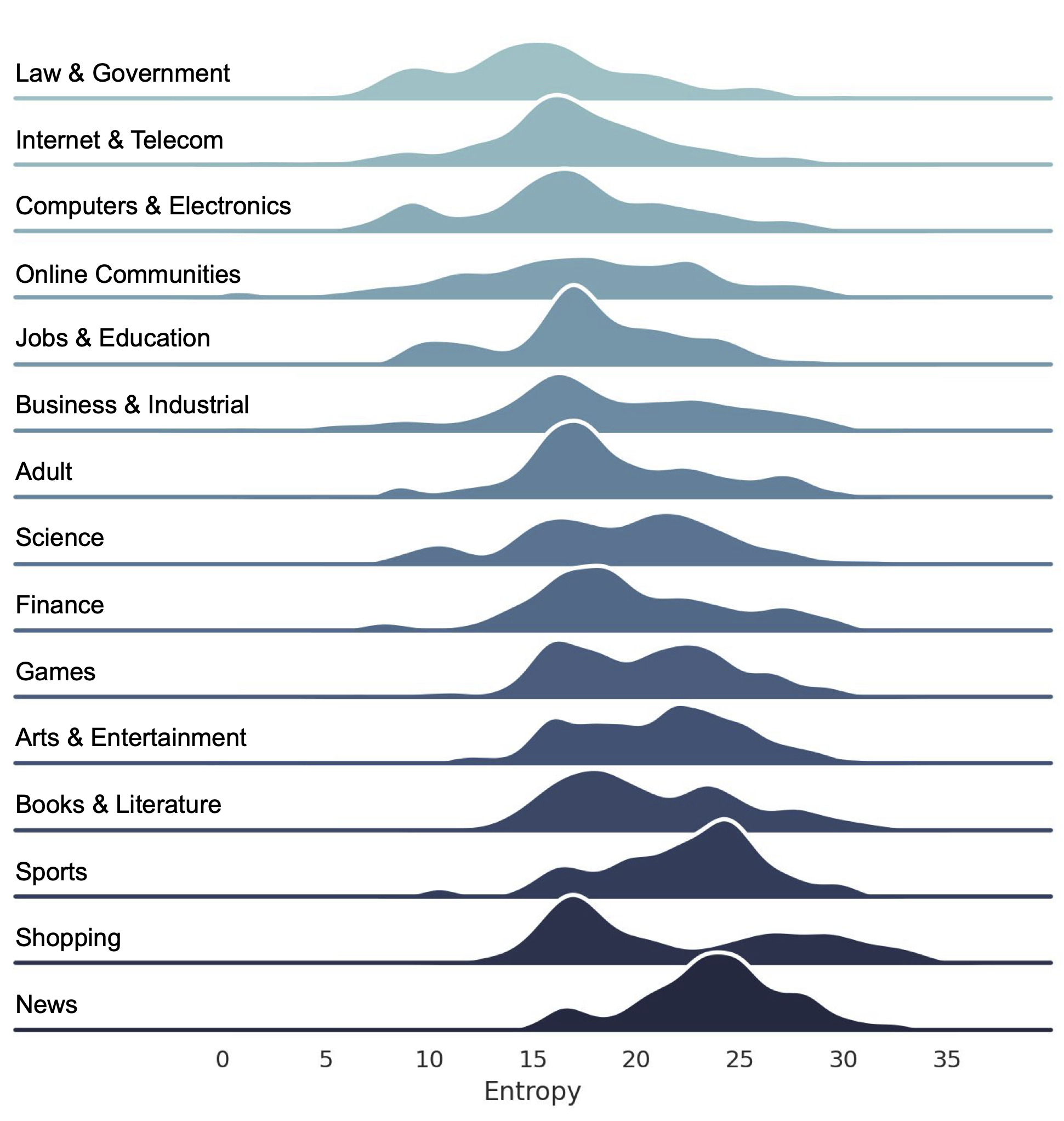}
  \caption{Per-vertical distributions of session entropy, sorted in increasing order of mean of the distribution.}
  \label{figure:verticals}
\end{figure}

Computing the correlation structure of surfaces also enabled us to estimate the joint entropy of all surfaces observed during any given session (recall the definition of a session from Section \ref{sec:notation}). Specifically, if $S$ is the set of surfaces observed during a session, then we estimate $H(X_S)$ using the Chow-Liu decomposition in Eq.~\eqref{eq:chowliu}. Since these estimates account for the mutual information between pairs of surfaces, they are more realistic than the estimates provided by any previous measurement study. Figure~\ref{figure:web} shows the distribution of session entropy for the Tranco Top 10,000 websites \citep{pochat2018tranco}. Sessions from the same website contribute to multiple buckets of this distribution, but the total contribution is normalized to be the same as all other websites. Note that the entropy for nearly all websites is either far below or far above 20 bits, suggesting that the typical website can be classified as operating in either a high- or low-entropy mode.\footnote{For context, note that if surface values are uniformly distributed in a population of 1 billion users, then (roughly) 10 bits of entropy corresponds to every user sharing their values with 1 million users, 20 bits to 1 thousand users, and 30 bits to every user being uniquely identified.}

Figure~\ref{figure:verticals} decomposes the session entropy distribution into several distributions according to each website's vertical. The distributions are normalized in the same way as the overall distribution in Figure \ref{figure:web}, but together they cover only 7187 websites instead of 10,000, since not all websites are assigned to a vertical, and we dropped verticals with fewer than 50 websites. Examining these distributions reveals clear differences between many of the verticals. For example, news and shopping websites have significantly higher entropy consumption than law \& government websites.

\begin{figure}[t]
  \centering
  \includegraphics[width=0.9\linewidth, trim={0 1.5cm 0 0}]{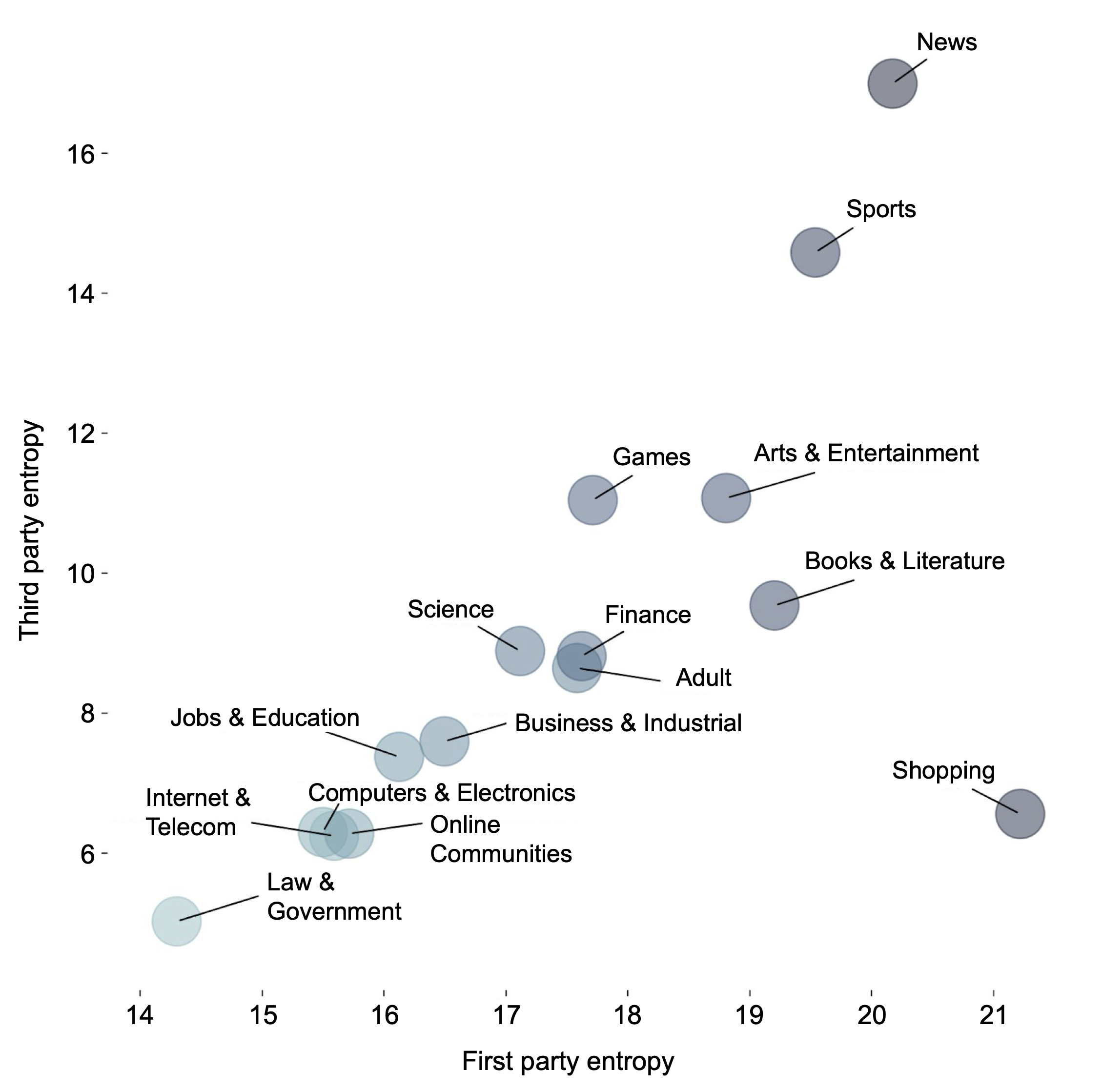}
  \caption{Average third-party versus first-party entropy}
  \label{figure:1p3p}
\end{figure}

We can obtain additional insight into differences among verticals by examining whether entropy is exposed in a first- or third-party context, which depends on the domain of the document that observes the surface values. A document in an iframe is first-party if it is same-site (i.e. shares the same registrable domain) with the top level URL (i.e.~the URL which the user is visiting), and is third-party otherwise.\footnote{Note that, even in a first-party context, the information collected via high-entropy APIs can be sent to a third-party (and even directly collected by a third-party script loaded in the first-party context), so this distinction is an approximation.}
Figure~\ref{figure:1p3p} plots the average first- and third-party entropy per vertical. A typical reason for embedding third-party iframes is to run advertising-related code, and so the differences between the entropy consumption patterns of verticals may be attributable to differences in business models. For example, news websites have the highest third-party entropy, possibly because of their heavy reliance on third-party advertising, while shopping websites have very low third-party entropy, and also seldom display third-party ads. Similarly, users of online communities (which include social networks) are often logged in, which decreases the need to fingerprint them, and this may explain why online community websites consume less overall entropy than news websites.

\subsection{Entropy as a fingerprinting metric}\label{section:validating-entropy}

While entropy is widely accepted as a proxy for fingerprinting risk, there has been limited data-driven evaluation of its suitability for that purpose. \citet{acar2014web}, based on earlier work by \citet{mowery2012pixel}, found that many fingerprinting libraries write specific and unusual strings to the canvas, because the way this string is rendered on a user’s device can be very identifying. Since these Web API function call `signatures' constitute clear evidence of known fingerprinting behavior, we checked whether they are associated with higher entropy. Figure~\ref{figure:fingerprintjs} is a box plot of the percentage of sessions with a fingerprinting signature versus their average session entropy. We observe a strongly positive correlation between the presence of signatures and entropy.

\begin{figure}[t]
  \centering
  \includegraphics[width=\linewidth]{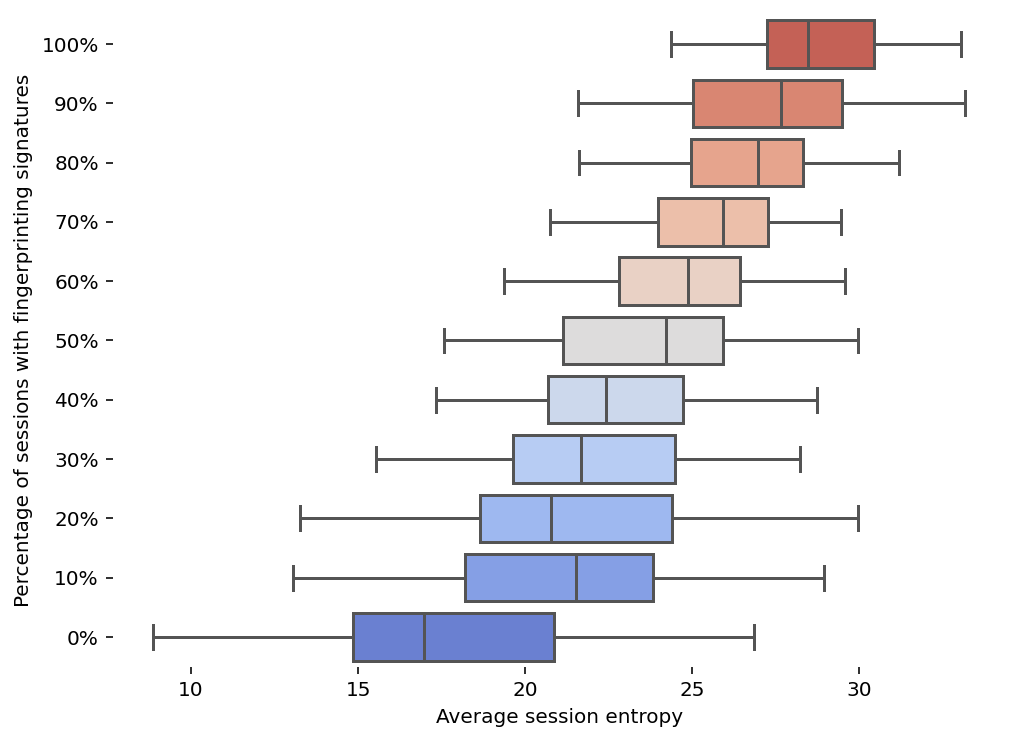}
  \caption{Fingerprinting signatures versus entropy}
  \label{figure:fingerprintjs}
\end{figure}

\subsection{Evaluating anti-fingerprinting methods}\label{section:chrome-blocklist}

\begin{figure}[t]
  \centering
  \includegraphics[width=\linewidth]{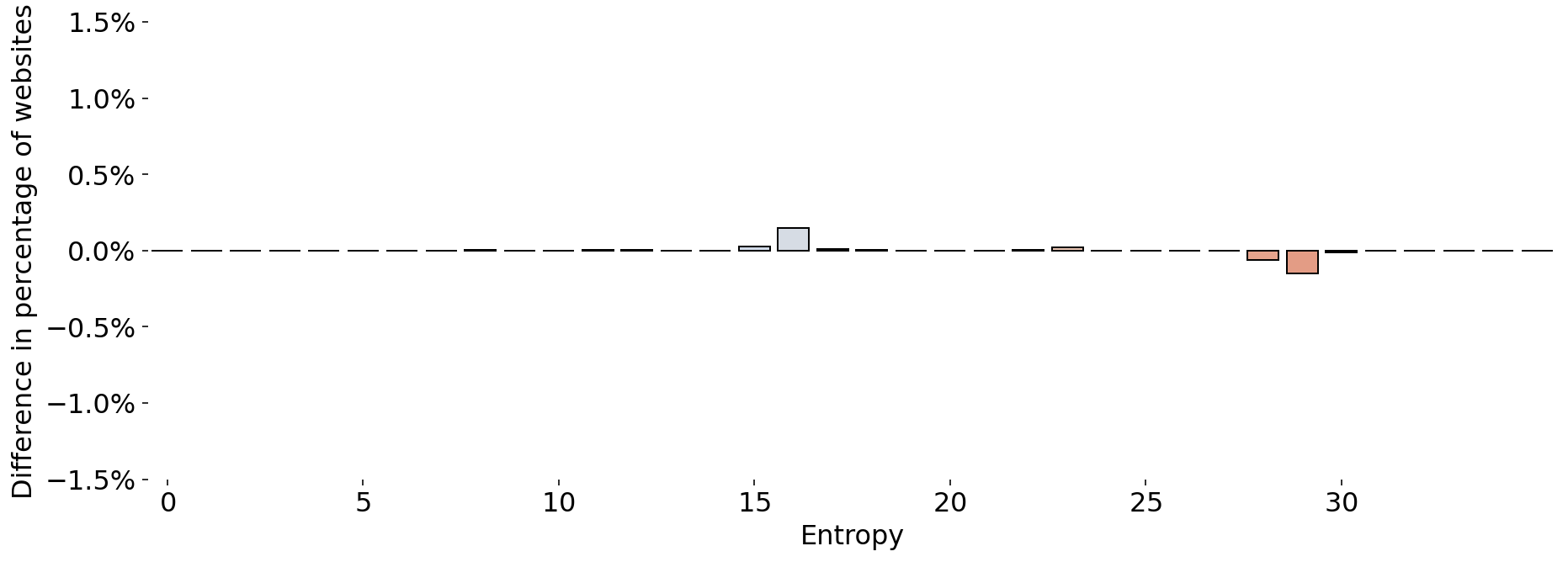}
  \caption{Change in entropy distribution when blocking scripts from known invasive fingerprinters \citep{disconnect}.}
  \label{figure:delta-blocking-disconnect-domains}
\end{figure}

\begin{figure}[t]
  \centering
  \includegraphics[width=\linewidth]{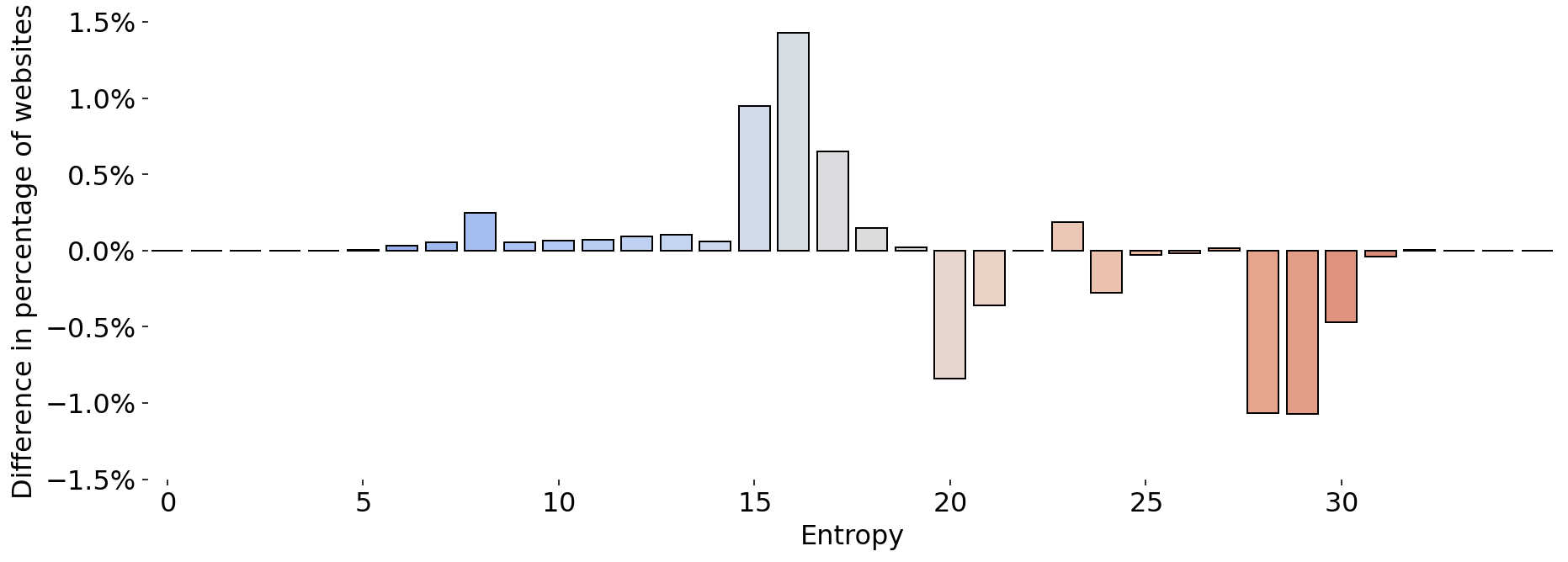}
  \caption{Change in entropy distribution when blocking scripts containing a fingerprinting signature \citep{acar2014web}.}
  \label{figure:delta-blocking-signatures}
\end{figure}

One approach to mitigating fingerprinting is for the browser to block the execution of scripts that are suspected of fingerprinting. Any such mitigation strategy must navigate a trade-off between limiting the amount of identifying information revealed to websites on the one hand, and disrupting users' experience of websites on the other hand. The results of our entropy measurement study can help quantify this trade-off. We instrumented a crawler to visit 6.4 million websites and record the arguments passed to every function called by every executed script, as well as the domain that hosts each script included in the page. We also calculated the entropy of the information that each website requested from the crawler during its visit.

We next simulated the impact of two fingerprinting mitigation strategies on this entropy distribution. Figure~\ref{figure:delta-blocking-disconnect-domains} shows the effect of blocking scripts that are hosted by a domain on the Disconnect list of invasive fingerprinters, which is used by some browsers as the basis of an anti-fingerprinting solution, and has been manually curated to minimize impact on valid website behavior \citep{disconnect}. Figure~\ref{figure:delta-blocking-signatures} shows the effect of the much more aggressive intervention of blocking every script containing at least one of the fingerprinting signatures found by \citet{acar2014web} (see Section \ref{section:validating-entropy} for further discussion of these signatures), irrespective of its impact on user experience. Note that Figures~\ref{figure:delta-blocking-disconnect-domains} and \ref{figure:delta-blocking-signatures} show the change in the entropy distribution relative to no intervention. Also, unlike in previous plots, the figures are not based on actual client sessions, but instead on crawler data, since for privacy reasons we did not record the detailed script execution history of any client belonging to a real user. We see that both interventions shift the entropy distribution to lower entropy, but the effect of the more aggressive intervention is much more significant. Browser developers can use this information to design their preferred intervention.

\section{Caveats}

Our results should be understood in light of the several limitations that had to be imposed on our measurement study, some for practical reasons, and others more fundamental. For the sake of both privacy and performance, we intentionally restricted the scope of our data collection. We examined only a subset of all surfaces, included only a small percentage of Chrome clients in the study, and capped the amount of information reported per client. All of these limits may have biased our results in ways that are difficult to quantify. For example, we only collected data from Chrome users who were logged in and consented to sharing their data with Google, a user population that may differ significantly from the overall population of the web.

Our study focused on measuring entropy, and while entropy is a reasonable proxy for fingerprinting risk (or, in other words, for the potential for fingerprinting), a website that collects high-entropy information is not necessarily a cause for user concern. The website may not be using the information for constructing fingerprints, or may be using the fingerprints for a benign purpose, such as preventing fraud or abuse. On the other hand, a website that collects low-entropy information may nonetheless be engaged in covert cross-site tracking via a non-fingerprinting method, such as link decoration or pixel tags.

Even our entropy estimates may be distorted, since they are based on the Chow-Liu decomposition, which provides an upper bound on entropy, but not a lower bound. As we discussed in Section \ref{sec:entropy}, this limitation is somewhat unavoidable, as computing more accurate entropy estimates would require a prohibitive number of samples. More technically, using a decomposition to estimate entropy while collecting data in a passive manner may have biased our estimates in a subtle way, since each term in the decomposition might have been estimated on a different set of clients.

\section{Conclusions}\label{section:conclusions}

Browser developers have proposed or implemented several mechanisms that try to address browser fingerprinting, including masking, fuzzing or coarsening the information the browser provides to websites (e.g., Brave, Firefox and Safari) and offering users the option to block domains known to engage in fingerprinting (e.g., Firefox and Edge). However, it remains unclear the extent to which such approaches can limit fingerprinting risks in-the-wild. Our measurement study is, to the best of our knowledge, the first effort to realistically quantify the amount of user-identifying information collected by websites that account for a large fraction of traffic on the Internet. The results of our study are useful for evaluating the impact of anti-fingerprinting methods, as well as for guiding the design of new anti-fingerprinting proposals. 

\newpage

\bibliographystyle{plainnat}
\bibliography{sample-base}

\newpage

\appendix

\section{Sample complexity of entropy estimation}
\label{sec:samplecomplexity}

The Chow-Liu (CL) upper bound depends on the entropy of marginal distributions $H(X_s)$ and on the mutual information between marginal distributions $I(X_s; X_{s'})$. The mutual information for discrete distributions can be further decomposed as
\begin{align}
I(X_s; X_{s'}) = H(X_s) + H(X_{s'}) - H(X_s, X_{s'}) \label{eq:mi_decomp}
\end{align}
which implies that computing the CL bound boils down to estimating entropy of marginal and pairwise marginal distributions. Therefore we will focus on estimating the entropy of discrete random variable $X$ over domain $\{1, \dots, k \}$ and its confidence interval which readily implies a confidence interval for mutual information and thus for the CL upper bound.

\subsection{Entropy estimation}
\label{sec:entropy_conf_int}

Assume that we are given $n$ realizations, denoted by $x_1, \dots, x_n$, from $X$ with parameters $\mathbf{p} = (p_1, \dots, p_k )$. This set of observation can be used to obtain a plug-in estimator $X^{(n)}$ of the distribution of $X$ as
\[
\widehat{p}^{(n)}_i = \frac{1}{n} \sum_{j=1}^n \mathbf{1}_{x_j = i} \enspace .
\]
In addition, the plug-in estimator can be naturally used to estimate the entropy as $H(X^{(n)})$.

\begin{theorem}\label{thm:entropy_bias}
Using the notation above, for any $\epsilon > \frac{2(k-1)}{n}$ and $n \in \mathbb{N}_+$, it holds that
\begin{align*}
\Pr \left( H(X) - H( X^{(n)} ) > \epsilon \right) & \le \exp \left(-\frac{\epsilon^2 n^2}{2\sum_{i=1}^k \log^2 p_i} \right) \\ & \quad + \exp\left(-\frac{n\epsilon}{2}\right)  \cdot \left( \frac{e\epsilon n}{2(k-1)}\right)^{k-1}
\end{align*}
and for any $\epsilon > 0$ and $n \in \mathbb{N}_+$, it holds that
\[
\Pr \left( H (X^{(n)} ) - H(X)  > \epsilon \right) \le \exp \left(-\frac{2\epsilon^2 n^2}{\sum_{i=1}^k\log^2 p_i} \right) \enspace .
\]
\end{theorem}

\begin{proof}
Based on (4.1) of \citet{Paninski03}, we can rewrite the ML estimate of entropy as
\begin{align*}
    H( X^{(n)} ) = H(X) + \sum_{i=1}^k \left( p_i - \widehat{p}_i^{(n)} \right) \log p_i - d_{\text{KL}} \left(\widehat{\mathbf{p}}^{(n)}, \mathbf{p} \right)
\end{align*}
which implies that
\begin{align}
    \Pr & \left( H(X) - H ( X^{(n)} ) > \epsilon \right) 
     \\ 
    & = \Pr \left( d_{\text{KL}} \left(\widehat{\mathbf{p}}^{(n)}, \mathbf{p} \right) - \sum_{i=1}^k \left( p_i - \widehat{p}_i^{(n)} \right) \log p_i  > \epsilon \right) \notag \\
    & \le 
    \Pr \left( d_{\text{KL}} \left(\widehat{\mathbf{p}}^{(n)}, \mathbf{p} \right) > \epsilon/2 \right) \notag 
    \\ & \qquad \qquad + \Pr \left( \sum_{i=1}^k p_i - \widehat{p}_i^{(n)} \log p_i  < - \epsilon/2 \right) \label{eq:decomp}
\end{align}
which follows from the union bound. The first term of (\ref{eq:decomp}) can upper bound by Theorem \ref{thm:conc_kl}. 
Using McDiarmid’s inequality, see for example Proposition 1 of \citet{AnKo01}, it holds that
\[
\Pr \left( \sum_{i=1}^k p_i - \widehat{p}_i^{(n)} \log p_i  < - \epsilon/2 \right) \le \exp \left( - \frac{\epsilon^2 n^2}{2\sum_{i=1}^k \log^2 p_i} \right)
\]
since the L1 sensitivity of $\sum_{i=1}^k p_i - \widehat{p}_i^{(n)} \log p_i$ is $\frac{\log p_i }{n}$ for all $i\in [k]$ and $n$.
On the other side, we have
\[
\Pr \left( H( X^{(n)} ) - H(X)  > \epsilon \right) \le \Pr \left( \sum_{i=1}^k p_i - \widehat{p}_i^{(n)} \log p_i  > \epsilon \right)
\]
which concludes the proof.
\end{proof}

\begin{theorem}[Theorem 1.2 of \citet{Agrawal20}]\label{thm:conc_kl} 
Using the notation defined above, for all $\epsilon > \frac{k-1}{n}$, it holds that
\[
\Pr \left( d_{\text{KL}} \left(\widehat{\mathbf{p}}^{(n)}, \mathbf{p} \right) \ge \alpha \right) \le e^{-n\alpha} \left( \frac{e\alpha n}{k-1}\right)^{k-1}
\]
where 
\[
d_{\textrm{KL}} ( \mathbf{p}, \mathbf{p}' ) = \sum_{i=1}^k p_i \log \frac{p_i}{p'_i} \enspace .
\]
\end{theorem}

One can reformulate Theorem \ref{thm:entropy_bias} as follows.
\begin{corollary}\label{corr:entropy_est}
For any $n>0$ and $\delta\in (0,1]$, it holds with probability at least $1-\delta$ that 
\[
H(X) - H( X^{(n)} ) \le \log \left(1+ \frac{k-1}{n} \right) + \sqrt{\frac{2\log (2/\delta) \log n}{n}}
\]
and
\[
H( X^{(n)} )- H(X) \le \sqrt{\frac{2\log (2/\delta) \log n}{n}}
\]
\end{corollary}

\subsection{Estimating mutual information}
\label{sec:mi_conf_int}

Mutual information between discrete random variables can be decomposed into entropy values as it is given in (\ref{eq:mi_decomp}). This result and Corollary \ref{corr:entropy_est} readily implies the following result.

\begin{theorem}\label{thm:mutual_info}
Suppose a joint distribution $(X_1, X_2)$ over domain $[k_1] \times [k_2]$ for which the marginal distributions $X_1$ and $X_2$ have parameters $\mathbf{p}_1$ and $\mathbf{p}_2$. 
Then, for any $\delta > 0$ and $n \in \mathbb{N}_+$, it holds that
\begin{align*}
I(X_1; X_2) & -  I\Big( X_1^{(n)} ; X_{2}^{(n)} \Big)  \le \\ & \log \left(1+ \frac{k_1-1}{n} \right) + \log \left(1+ \frac{k_2-1}{n} \right) + 3\sqrt{\frac{6\log 2/\delta \log n}{n}}
\end{align*}
and 
\begin{align*}
I\Big( X_1^{(n)} ; X_{2}^{(n)} \Big) - I(X_1; X_2 &)  \le \\ & \log \left(1+ \frac{k_1\cdot k_2 -1}{n} \right)   + 3\sqrt{\frac{6\log 2/\delta \log n}{n}}
\end{align*}
\end{theorem}

Based on Theorem \ref{thm:mutual_info}, note that the confidence interval for mutual information consists of the upper bound of the bias term and the statistical error. The statistical error becomes smaller than 0.5 with $n=30000$ and $\delta=0.05$. The bias term also becomes smaller than 0.5 when $n> 3 k_1 \cdot k_2$. Thus if our goal is to achieve an 1-bit additive error of the mutual information with $\delta = 0.05$, we need 
\[
n=\max \left( 3k_1\cdot k_2, 30000 \right)
\]
samples. 
So if the domain size small, i.e. $k_1 \cdot k_2 < 10000$, additive error is dominated by statistical error term, and the domain size is larger than we need more sample to control the magnitude of the bias. 

\subsection{Estimating domain sizes}

Note that the confidence intervals for the entropy and mutual information of random variables given in Sections \ref{sec:entropy_conf_int} and \ref{sec:mi_conf_int} depend on their domain size (denoted $k$, $k_1$ and $k_2$ in the corresponding expressions). When each random variable represents a surface, its domain size is the number of distinct values for the surface among all clients. Since this information is unknown for most surfaces, we substituted the number of \emph{observed} surface values instead. However, this quantity can be an underestimate for the true number of distinct surface values when the sample size is small, so we only used it when the number of samples exceeded the number of observed surface values by a factor of at least 30, a conservative heuristic. Otherwise we omitted all the mutual information terms associated with the surface when computing the Chow-Liu upper bound (Eq.~\eqref{eq:chowliu}), which loosens the upper bound but does not impact its validity.

\section{Minimizing observable impact of the experiment}
\label{sec:observable}

When estimating the fingerprinting risk posed by Web APIs, it can happen that the estimation process itself creates a novel fingerprinting surface. For example, as we explained in Section \ref{sec:methodology}, the second and third phases of our experiment assigned a different list of surfaces $L_c$ to each client $c$. If our experiment had caused clients to record the values of surfaces in $L_c$ while leaving the values of surfaces not in $L_c$ unrecorded, then a website might be able infer $L_c$ by measuring which surfaces take slightly longer to execute on client $c$ then expected. Since $L_c$ varies from client to client, it represents a fingerprinting risk. We mitigated this risk by having each client record the value of every surface in the experiment, whether or not the surface was contained in $L_c$. The recorded surface values were later filtered to only include the surfaces in $L_c$, and then asynchronously reported by the client to a server. 

\section{Special handling of codec, canvas readback and font surfaces}
\label{sec:special}

Since each surface consists of a Web API and possible input parameters, some Web APIs encompass so many surfaces that it is impractical to estimate their entropy even using the Chow-Liu decomposition (Eq.~\eqref{eq:chowliu}). Instead we took a tailored approach to each of these cases. For media codecs, we found that nearly every surface with non-zero entropy was highly correlated with the platform of the device (\emph{i.e.}, Android, Windows, \emph{etc.}), and so platform entropy was a reasonable proxy for their joint entropy. For canvas readback surfaces, we estimated the entropy of each of the 50 most common canvas readback operations, and conservatively used the highest of these entropies as a proxy for the entropy of any other canvas readback operation. For fonts surfaces, we used a formula inspired by previous work \citep{eckersley2010unique, fifield2015fingerprinting}.

\section{Greedy surface subset assignment}
\label{sec:greedy}

As explained in Section \ref{sec:methodology}, in the third phase of our experiment we assigned to each client $c$ a subset of surfaces $L_c$, with the goal of collecting a sufficient number of samples from all pairs of the 60 `core' surfaces. We designed the assignment to satisfy two constraints. First, each surface pair $\{i, j\}$ had to be assigned to at least $n_{ij}$ clients, where $n_{ij}$ is the minimum number of samples needed to accurately estimate the mutual information between surfaces $i$ and $j$, as determined by the derivation is Appendix \ref{sec:mi_conf_int}. Second, the joint entropy of the surfaces in any subset could not exceed a given threshold. We conservatively estimated the joint entropy of a subset using the sum of the marginal entropies of the surfaces in the subset, where each marginal entropy was estimated from data collected in the second phase of the experiment.

We used Algorithm \ref{alg:block} below to construct an assignment that respects the two constraints described above while also attempting to minimize the number of clients in this phase of the experiment. In each iteration, the algorithm greedily selects a new surface subset $L^*$ that satisfies the sample requirements for surface pairs that were not met by any previously selected subset, while also respecting the joint entropy constraint. The subset $L^*$ is then assigned to some clients that were not assigned a subset in a previous iteration.

\begin{algorithm}[!h]
\caption{Greedy surface subset assignment \label{alg:block}}
\begin{algorithmic}[1] 
\Statex {\bf Given:}
    \begin{itemize}
    \item Number of samples $n_{ij}$ needed for each surface pair $\{i, j\}$.
    \item Marginal entropy $h_i$ of each surface $i$.
    \item Upper bound $H$ on joint entropy of any surface subset.
    \end{itemize} 
\Statex
\Statex {\bf Define} for each surface subset $L$:
    \begin{itemize}
    \item Let $h(L) = \sum_{i \in L} h_i$ be the sum of marginal entropies of the surfaces in $L$.
    \item Let $s(L) = \sum_{\{{i, j\}} \subseteq L} n_{ij}$ be the total number of samples needed for all surface pairs in $L$.
    \item Let $m(L) = \max_{\{i, j\} \subseteq L} n_{ij}$ be the maximum number of samples needed for any surface pair in $L$.
    \end{itemize}
\Statex
\WHILE{at least one of the $n_{ij}$’s is non-zero}
    \State Find a surface subset $L^*$ that solves \[\max_{L:~ h(L) \le H} s(L).\]
    \State Select $m(L^*)$ clients that do not have a subset assignment.
    \State For each selected client $c$ let $L_c = L^*$. 
    \State For each $\{i, j\} \subseteq L^*$ let $n_{ij} = 0$.
\ENDWHILE
\end{algorithmic}
\end{algorithm}

\section{Replacing missing mutual information terms}
\label{sec:monotonicity}

Recall from Section \ref{sec:methodology} that if we do not have enough samples of a pair of surfaces $(X_i, X_j)$ then we omit the mutual information term $I(X_i; X_j)$ from the right-hand side of the Chow-Liu decomposition (Eq.~\eqref{eq:chowliu}). We are allowed to do this because mutual information is always at least zero. Alternatively, we can replace $I(X_i; X_j)$ with any valid lower bound. The next theorem gives a lower bound for $I(X_1; X_n)$ that can be significantly larger than zero provided that we have enough samples of the pairs $(X_1, X_2), (X_2, X_3), \ldots, (X_{n-1}, X_n)$. 

\begin{theorem} For any random variables $X_1, \ldots, X_n$
\[
I(X_1; X_n) \ge \sum_{i=1}^{n-1} I(X_i; X_{i+1}) - \sum_{i=2}^{n-1} H(X_i).
\]
\end{theorem}
\begin{proof} By the Chow-Liu decomposition (Eq.~\eqref{eq:chowliu})
\[
H(X_1, \ldots, X_n) \le \sum_{i=1}^n H(X_i) - \sum_{i=1}^{n-1} I(X_i; X_{i+1})
\]
and thus by rearranging
\[
H(X_1) + H(X_n) - H(X_1, \ldots, X_n) \ge \sum_{i=1}^{n-1} I(X_i; X_{i+1}) - \sum_{i=2}^{n-1} H(X_i).
\]
The theorem follows by noting that $H(X_1, \ldots, X_n) \ge H(X_1, X_n)$ and $I(X_1; X_n) = H(X_1) + H(X_n) - H(X_1, X_n)$.
\end{proof}

\section{Web APIs in the experiment}
\label{sec:webapis}

The following 161 Web APIs were selected for our experiment according to the criteria described in Section \ref{sec:methodology}.

\begin{lstlisting}[
    basicstyle=\ttfamily\footnotesize,
    numberstyle=\tiny,
    numbers=left,                    
    numbersep=5pt]
AudioBuffer.getChannelData
AudioContext.baseLatency
AudioDecoder.configure
BaseAudioContext.sampleRate
Document.plugins
FeaturePolicy.features
Gamepad.buttons
Gamepad.id
HTMLCanvasElement.getContext
HTMLMediaElement.canPlayType
InputDeviceCapabilities.firesTouchEvents
Keyboard.getLayoutMap
MediaCapabilities.decodingInfo
MediaCapabilities.encodingInfo
MediaDeviceInfo.kind
MediaDeviceInfo.label
MediaDeviceInfo.toJSON
MediaDevices.enumerateDevices
MediaDevices.getDisplayMedia
MediaDevices.getSupportedConstraints
MediaDevices.getUserMedia
MediaRecorder.mimeType
Navigator.appVersion
Navigator.bluetooth
Navigator.canShare
Navigator.cookieEnabled
Navigator.deviceMemory
Navigator.doNotTrack
Navigator.getBattery
Navigator.getGamepads
Navigator.getUserAgent
Navigator.getUserMedia
Navigator.hardwareConcurrency
Navigator.javaEnabled
Navigator.language
Navigator.languages
Navigator.maxTouchPoints
Navigator.mediaCapabilities
Navigator.mediaDevices
Navigator.mimeTypes
Navigator.platform
Navigator.plugins
Navigator.productSub
Navigator.userAgent
Navigator.userAgentData
Navigator.vendor
Navigator.webkitGetUserMedia
NetworkInformation.saveData
PaintWorkletGlobalScope.devicePixelRatio
Plugin.description
Plugin.filename
Plugin.name
RTCIceCandidate.address
RTCIceCandidate.candidate
RTCIceCandidate.port
RTCIceCandidate.relatedAddress
RTCIceCandidate.relatedPort
RTCPeerConnection.getReceivers
RTCRtpReceiver.getCapabilities
RTCRtpSender.getCapabilities
Screen.availHeight
Screen.availLeft
Screen.availTop
Screen.availWidth
Screen.colorDepth
Screen.height
Screen.id
Screen.internal
Screen.left
Screen.name
Screen.pixelDepth
Screen.primary
Screen.scaleFactor
Screen.top
Screen.touchSupport
Screen.width
ScreenOrientation.angle
ScreenOrientation.lock
ScreenOrientation.type
VisualViewport.height
VisualViewport.offsetLeft
VisualViewport.offsetTop
VisualViewport.pageLeft
VisualViewport.pageTop
VisualViewport.scale
VisualViewport.width
WebGLCompressedTextureASTC.getSupportedProfiles
WebGLRenderingContext.getBufferParameter
WebGLRenderingContext.getExtension
WebGLRenderingContext.getParameter
WebGLRenderingContext.getRenderbufferParameter
WebGLRenderingContext.getShaderPrecisionFormat
WebGLRenderingContext.getSupportedExtensions
WebGL2ComputeRenderingContext.getBufferParameter
WebGL2ComputeRenderingContext.getExtension
WebGL2ComputeRenderingContext.getFramebufferAttachmentParameter
WebGL2ComputeRenderingContext.getIndexedParameter
WebGL2ComputeRenderingContext.getInternalformatParameter
WebGL2ComputeRenderingContext.getParameter
WebGL2ComputeRenderingContext.getRenderbufferParameter
WebGL2ComputeRenderingContext.getShaderPrecisionFormat
WebGL2ComputeRenderingContext.getSupportedExtensions
WebGL2RenderingContext.getBufferParameter
WebGL2RenderingContext.getExtension
WebGL2RenderingContext.getIndexedParameter
WebGL2RenderingContext.getInternalformatParameter
WebGL2RenderingContext.getParameter
WebGL2RenderingContext.getRenderbufferParameter
WebGL2RenderingContext.getShaderPrecisionFormat
WebGL2RenderingContext.getSupportedExtensions
WebGL2RenderingContext.makeXRCompatible
WheelEvent.deltaMode
WheelEvent.deltaX
WheelEvent.deltaY
WheelEvent.deltaZ
Window.devicePixelRatio
Window.innerHeight
Window.innerWidth
Window.matchMedia
Window.name
Window.onbeforeunload
Window.ondevicemotion
Window.ondeviceorientation
Window.ondeviceorientationabsolute
Window.onlanguagechange
Window.onmessage
Window.onoffline
Window.ononline
Window.onorientationchange
Window.onscreenschange
Window.onsuspend
Window.ontimeupdate
Window.ontimezonechange
Window.orientation
Window.outerHeight
Window.outerWidth
Window.pageXOffset
Window.pageYOffset
Window.postMessage
Window.postMessage
Window.screenLeft
Window.screenTop
Window.screenX
Window.screenY
Window.scrollX
Window.scrollY
WorkerGlobalScope.onlanguagechange
WorkerGlobalScope.ontimezonechange
WorkerNavigator.appVersion
WorkerNavigator.deviceMemory
WorkerNavigator.hardwareConcurrency
WorkerNavigator.language
WorkerNavigator.languages
WorkerNavigator.mediaCapabilities
WorkerNavigator.platform
WorkerNavigator.serial
WorkerNavigator.usb
WorkerNavigator.userAgent
WorkerNavigator.userAgentData
CSS Media Features
Fonts
\end{lstlisting}

\end{document}